\documentclass[12pt]{amsart}

\usepackage{amsaddr}
\usepackage{amssymb,amsmath,amscd,amsthm,enumerate,verbatim,color,url}
\usepackage[utf8]{inputenc}
\usepackage[mathscr]{euscript}
\usepackage{bbm}
\usepackage[left=1in,right=1in,top=1in,bottom=1in]{geometry}

\definecolor{purple}{rgb}{.5,0,1}
\definecolor{orange}{rgb}{1,.5,0}
\definecolor{pink}{rgb}{1,0,.5}
\definecolor{green}{rgb}{0,.5,0}
\definecolor{gold}{rgb}{1,.623,0}

\def\trv#1{\textcolor{blue}{#1}}

\date{}

\newcommand{\A}{{\mathcal A}}

\newcommand{\Z}{{\mathbbm Z}}
\newcommand{\R}{{\mathbbm R}}
\newcommand{\C}{{\mathbbm C}}
\newcommand{\N}{{\mathbbm N}}
\newcommand{\T}{{\mathbbm T}}

\newcommand{\Q}{{\mathbbm Q}}
\newcommand{\D}{{\mathbbm D}}
\newcommand{\U}{{\mathbbm U}}

\newcommand{\CE}{{\mathcal E}}

\newcommand{\e}{{\varepsilon}}

\newcommand{\CL}{{\mathcal{L}}}
\newcommand{\CM}{{\mathcal{M}}}
\newcommand{\CZ}{{\mathcal{Z}}}

\newcommand{\pD}{{\partial \D}}

\newcommand{\Arg}{{\mathrm{Arg}}}

\newcommand{\Id}{{\mathbbm{I}}}

\newcommand{\Int}{{\mathrm{int}}}

\newcommand{\dd}{{\mathrm{d}}}

\newcommand{\ac}{{\mathrm{ac}}}

\newcommand{\SU}{{\mathbbm{SU}}}

\newcommand{\Hd}{{\mathrm{H}}}

\DeclareMathOperator*{\slim}{s-lim}

\newtheorem{theorem}{Theorem}[section]
\newtheorem{lemma}[theorem]{Lemma}
\newtheorem{prop}[theorem]{Proposition}
\newtheorem{coro}[theorem]{Corollary}

\theoremstyle{definition}

\newtheorem*{remark*}{Remark}
\theoremstyle{definition}

\theoremstyle{definition}

\theoremstyle{definition}

\allowdisplaybreaks

\numberwithin{equation}{section}

\renewcommand{\Im}{\mathrm{Im}}
\renewcommand{\Re}{\mathrm{Re}}
\newcommand{\tr}{\mathrm{tr} }

\newcommand{\Leb}{{\mathrm{Leb}}}

\newcommand{\set}[1]{\left\{#1\right\}}
\newcommand{\eqdef}{\overset{\mathrm{def}}=}

\begin{document}

\title[Spectral Approximation for CMV Matrices]{Spectral Approximation for Ergodic CMV Operators with an Application to Quantum Walks}

\author{Jake Fillman}
\address{Virginia Tech\\
Mathematics (MC0123) \\
225 Stanger Street \\
Blacksburg, VA 24061 \\
USA}
\email{fillman@vt.edu}

\author{Darren C. Ong}
\address{Xiamen University Malaysia,\\
Jalan Sunsuria, Bandar Sunsuria,\\
43900 Sepang, Selangor Darul Ehsan,\\
Malaysia}
\email{darrenong@xmu.edu.my}

\author{Tom VandenBoom}
\address{Rice University\\
6100 Main St., MS-136\\
Houston, TX 77005\\
USA}
\email{tv4@rice.edu}

\maketitle

\begin{abstract}
We establish concrete criteria for fully supported absolutely continuous spectrum for ergodic CMV matrices and purely absolutely continuous spectrum for limit-periodic CMV matrices. We proceed by proving several variational estimates on the measure of the spectrum and the vanishing set of the Lyapunov exponent for CMV matrices, which represent CMV analogues of results obtained for Schr\"odinger operators due to Y.\ Last in the early 1990s. Having done so, we combine those estimates with results from inverse spectral theory to obtain purely absolutely continuous spectrum.
\end{abstract}

\section{Introduction}

The motivation for this paper came about as one of the authors was writing \cite{F2017CMP}; there had been a substantial amount of recent activity on the connection between various versions of ballistic motion and purely absolutely continuous (a.c.) spectrum for 1-dimensional operators (Schr\"odinger, Jacobi, and CMV) \cite{AK98, DLY2015CMP, F2017CMP, Kachkovskiy2016CMP, Zhao2017JDE, ZhangZhao2017CMP}. Since the methods of \cite{F2017CMP} produce ballistic motion for limit-periodic CMV matrices satisfying a Pastur--Tkachenko-like condition, we wanted to verify that such CMV matrices indeed had purely a.c.\ spectrum. Along the way, we realized that the proof of a.c.\ spectrum could be accomplished in a rather elegant manner by making use of spectral approximation results analogous to those known for Schr\"odinger operators, but which were as-yet unknown for CMV matrices. The main aim of this paper is to work out the appropriate CMV analogues of those approximation results and use them to deduce purely a.c.\ spectrum in the CMV analogue of the Pastur--Tkachenko class. In this application, we know that the spectrum is homogeneous in the Carleson sense, which is how we are able to deduce pure a.c.\ spectrum. In general, one needs to know at least that the spectrum has positive Lebesgue measure in order to deduce nontrivial conclusions from the approximation results, so, as a supplement to our work, we also establish a criterion to guarantee positive-measure spectrum for limit-periodic CMV matrices. To the best of our knowledge, this condition that ensures positive-measure spectrum is new, even for Schr\"odinger or Jacobi operators.

An \emph{extended CMV matrix} is a pentadiagonal unitary operator on $\ell^2(\Z)$ with a repeating $2 \times 4$ block structure of the form
\begin{equation} \label{eq:stdcmvdef}
\CE
=
\begin{bmatrix}
\ddots & \ddots & \ddots & \ddots &&&&& \\
& \overline{\alpha_0}\rho_{-1} 
& -\overline{\alpha_0}\alpha_{-1} 
& \overline{\alpha_1} \rho_0 
& \rho_1\rho_0 &&&\\
& \rho_0\rho_{-1} 
& -\rho_0\alpha_{-1} 
& -\overline{\alpha_1}\alpha_0 
& -\rho_1\alpha_0 & &&\\
&&& \overline{\alpha_2}\rho_1 
& -\overline{\alpha_2}\alpha_1 
& \overline{\alpha_3} \rho_2 
& \rho_3\rho_2 &\\
&&& \rho_2\rho_1 
& -\rho_2\alpha_1 
& -\overline{\alpha_3}\alpha_2 
& -\rho_3\alpha_2 & \\
&&&& \ddots & \ddots & \ddots & \ddots
\end{bmatrix},
\end{equation}
where $\alpha_n \in \D \eqdef \{ z \in \C : |z| < 1\}$ and $\rho_n \eqdef \sqrt{1-|\alpha_n|^2}$ for all $n \in \Z$. Such operators arise naturally in several settings, including quantum walks in one dimension \cite{CGMV} and the one-dimensional ferromagnetic Ising model \cite{DFLY2015IMRN, DMY2013JSP}. Moreover, half-line CMV matrices, obtained by setting $\alpha_{-1} = -1$ and restricting to $\ell^2(\Z_+)$ arise naturally in the study of orthogonal polynomials on the unit circle (OPUC) \cite{S1, S2} and moreover are universal within the class of unitary operators having a cyclic vector (in the sense that any unitary operator with a cyclic vector is unitarily equivalent to a half-line CMV matrix).

We will be particularly interested in the case in which the coefficients of $\CE$ are generated by an underlying dynamical system. Concretely, given a Borel probability space $(\Omega,\mu)$, a measurable $\mu$-ergodic map $T:\Omega \to \Omega$ (with measurable inverse), and a measurable map $f:\Omega \to \D$, one can consider $\CE^\omega$ for $\omega \in \Omega$ with coefficients given by
\[
\alpha_n^\omega 
:= 
f(T^n\omega),
\quad
n \in \Z.
\]
One can then study $\{ \CE^\omega\}_{\omega \in \Omega}$ as a family and, by so doing, leverage tools and techniques from dynamical systems to prove statements about $\CE^\omega$ for $\mu$-almost every $\omega$ (or even every $\omega$ in some situations). This scheme subsumes many particular cases under one umbrella, such as those for which $\alpha_n$ is almost-periodic (in the sense of Bohr or Bochner). In that case, $\Omega$ is a compact monothetic group and $T$ is a translation by a topological generator of $\Omega$. Some common examples of almost-periodic operators include:
\begin{itemize}
\item $\Omega = \Z_p$, $T:\omega\mapsto\omega+1$: The resulting operators are \emph{periodic}, and one can compute spectra and spectral data more-or-less explicitly \cite[Chapter~11]{S2}.

\item $\Omega = \T^d := \R^d / \Z^d$, $T:\omega \mapsto \omega+\beta$, where $\{1,\beta_1,\beta_2,\ldots,\beta_d\}$ is linearly independent over $\Q$, and $f$ is continuous. The resulting operators are \emph{quasi-periodic}.

\item $\Omega$ is a procyclic group (e.g.\ $\Omega = \mathbbm{J}_p$, the $p$-adic integers for $p$ prime). The resulting operators are \emph{limit-periodic} (for more on the connection between limit-periodic operators and procyclic groups, see~\cite{Avila2009CMP,Gan2010}).

\end{itemize}

The \emph{Lyapunov exponent}, 
\begin{equation} \label{eq:lyapdef}
L(z)
:=
\lim_{n \to \infty} \frac{1}{n} \int_\Omega \log \| S(f(T^{n-1} \omega),z) \cdots S(f(T\omega),z) S(f(\omega),z) \| \, \dd \mu(\omega),
\quad
z \in \pD
\end{equation}
plays a key role in the spectral analysis of the family $\{\CE^\omega\}_{\omega \in \Omega}$. Here, $S$ denotes the so-called \emph{Szeg\H{o} transfer matrix}, defined by
\begin{equation} \label{eq:szegotmdef}
S(\alpha,z)
:=
\frac{1}{\sqrt{1-|\alpha|^2}}
\begin{bmatrix}
z & - \overline{\alpha} \\
-z\alpha & 1
\end{bmatrix},
\quad
\alpha \in \D, \; z \in \C.
\end{equation}
Because the matrices so defined satisfy $|\det S(\omega,z)| = 1$ for $z \in \pD$, we have $L(z) \ge 0$ for $z \in \pD$. We define
\[
\CZ
=
\set{z \in \pD : L(z) = 0}.
\]
By general considerations, there is a compact set $\Sigma = \Sigma(f)$ with $\sigma(\CE^\omega) = \Sigma$ for $\mu$-a.e.\ $\omega \in \Omega$. Moreover, by the Combes--Thomas estimate, one has $\CZ \subseteq \Sigma$. We are interested in spectral approximation of CMV operators, so the following questions are natural:
\begin{itemize}
\item To what extent does $\CZ$ approximate $\Sigma$?
\item The same question, but with $\Sigma$ replaced by spectra of suitable periodic approximants of $\CE^\omega$.
\end{itemize}
In particular, helpful answers to both questions can allow one to bootstrap information about spectra of periodic approximants into information about the spectra of more exotic ergodic operators, using $\CZ$ as an intermediary.

\subsection{Results}

The set $\CZ$ is well-approximated by spectra of periodic approximants. Concretely, we define $\widetilde{\CE}^{\omega,q}$ for $\omega \in \Omega$ and $q \in \Z_+$ by
\begin{equation}\label{eq:widetildealpha}
\widetilde\alpha_{j + nq}^{\omega,q}
=
\alpha_j^\omega
\text{ for all } 0 \le j \le q-1 \text{ and } n \in \Z.
\end{equation}
We then denote $\Sigma^{\omega,q} = \sigma(\widetilde{\CE}^{\omega,q})$.

\begin{theorem} \label{t:perapprox}
For $\mu$-almost every $\omega \in\Omega$, we have
\[
\Leb\left( \limsup_{q \to \infty} \Sigma^{\omega, q}\setminus \CZ \right)
=
0.
\]
\end{theorem}

In fact, we only ever work with periodic approximants of even periods, so we prove the slightly stronger statement
\[
\Leb\left( \limsup_{q\to\infty} \Sigma^{\omega,2q}\setminus \CZ \right)
=
0.
\]
\medskip

Next, we discuss some applications to limit-periodic CMV matrices. We say that a CMV matrix is \emph{periodic} if there exists $q \in \Z_+$ with the property that $\alpha_{n+q} \equiv \alpha_n$ for all $n$; a \emph{limit-periodic} CMV is one which may be obtained as an operator-norm limit of periodic CMV matrices, i.e., $\CE$ is limit-periodic if there exist periodic CMV matrices $\CE_1$, $\CE_2$, \dots with
\begin{equation} \label{eq:lpCMVdef}
\lim_{n \to \infty} \| \CE_n - \CE \|
=
0.
\end{equation}
When $\CE$ is limit-periodic with periodic approximants $\CE_n$, we will write $\Sigma = \sigma(\CE)$ and $\Sigma_n = \sigma(\CE_n)$. Our next result applies the foregoing results to deduce a concrete criterion for positive-measure a.c. spectrum.

\begin{theorem}
\label{t:LPposmeas}
Let $\CE$ be a limit-periodic CMV matrix with Verblunsky coefficients uniformly bounded away from the circle $\|\alpha\|_\infty < 1$ and $q_n$-periodic approximants $\CE_n$ having Verblunsky coefficients $\alpha^{(n)}$, and suppose
\begin{equation} \label{eq:alphao(q)}
\lim_{n\to\infty} q_n\|\alpha^{(n)} - \alpha\|_\infty 
=
0.
\end{equation}
Then,
\begin{align}
\label{eq:fullZ}
\Leb(\Sigma \setminus \CZ) = 0.
\end{align}
Additionally, if
\begin{align}
\label{eq:LPsumbd}
\sum_{n=k+1}^\infty q_n\|\CE_n - \CE_{n-1}\| < \frac{1}{2} \Leb(\Sigma_k)
\end{align}
for some $k$, then \eqref{eq:alphao(q)} and \eqref{eq:fullZ} hold and we have $\Leb(\Sigma) > 0$.
\end{theorem}

Let us point out that the analogue for this criterion is new even for discrete Schr\"odinger operators, and substantially less restrictive than that of \cite[Theorem~4.ii]{Last1992CMP}. In fact, the criterion of \cite{Last1992CMP} implies not just positive measure, but \emph{homogeneity} of the spectrum \cite{FilLuk2017JST}. Once we put everything together, we will see that homogeneity plus variational estimates yield pure absolutely continuity of the spectra of limit-periodic CMV matrices as long as the amplitudes of the low-frequency modes decay sufficiently rapidly. To make things more specific, given a monotone decreasing function $\phi: \R_+ \to \R_+$ with $\phi(n) \to 0$ as $n \to \infty$, let us say that $\CE$ is limit-periodic with \emph{rate function} $\phi$ if the limit on the left-hand side of \eqref{eq:lpCMVdef} converges like $\phi(q_n)$, i.e.
\[
\| \CE_n -\CE\|
\leq
\phi(q_n).
\]
Common choices for $\phi$ include polynomial decay ($\phi(x) = x^{-p}$, $p>0$) and exponential decay ($\phi(x) = e^{-ax}$, $a>0$).

In general, it is a very interesting question to what extent one can characterize the ``phase transition'' undergone by limit-periodic operators. Concretely, as the rate of approximation by periodic operators (quantified by $\phi$) grows slower, one can obtain increasing singularity of the spectrum and spectral measures: e.g.\ purely singular continuous spectrum on a set of zero Hausdorff dimension \cite{Avila2009CMP} or even purely pure point spectrum (in the discrete Schr\"odinger operator setting) \cite{DamGor2016AdvM}. Moreover, many of the phenomena for limit-periodic operators (considered as a class of operators) seem to be somewhat universal; for example, purely singular continuous spectrum supported on a set of Hausdorff dimension zero occurs for discrete Schr\"odinger operators \cite{Avila2009CMP}, CMV matrices \cite{FilOng2017JFA} and continuum Schr\"odinger operators \cite{DamFilLuk2042JST}.

Combining Theorem \ref{t:LPposmeas} with some inverse spectral results and homogeneity of the spectrum, we can prove purely absolutely continuous spectrum for the Pastur--Tkachenko class. Specifically, we say that $\CE$ is in the Pastur--Tkachenko class if its Verblunsky coefficients are bounded uniformly away from $\pD$ and it admits $q_n$-periodic approximants $\CE_n$ with
\[
\lim_{n\to\infty} e^{\eta q_{n+1}} \|\alpha^{(n)} - \alpha\|_\infty
=
0
\text{ for every } \eta  > 0.
\]
\begin{theorem} \label{t:PTac}
Any $\CE$ in the Pastur--Tkachenko class has purely absolutely continuous spectrum.
\end{theorem}

That Theorem \ref{t:PTac} holds in the Jacobi \cite{Egorova1987} and Schr\"odinger \cite{PasTka1984,PasTka1988} settings has been well-established. However, our proof is not an analogue of the arguments from those references; instead, by using our spectral approximation results, our proof works by ``bootstrapping'' the very weak result $\Leb(\Sigma \setminus \CZ)=0$ from Theorem~\ref{t:LPposmeas} into pure a.c.\ spectrum by using results from the inverse theory. It is tempting to think that one can immediately proceed from $\Leb(\Sigma \setminus \CZ)=0$ to purely a.c.\ spectrum via Kotani Theory, but there are at least two obstacles:
\begin{itemize}
\item If $\Leb(\Sigma) = 0$, then $\Leb(\Sigma \setminus \CZ)=0$ holds trivially, but $\Sigma$ cannot support absolutely continuous measures, and hence the a.c.\ spectrum is empty in this case.
\smallskip

\item Even when $\Leb(\Sigma) > 0$, the strongest conclusion that one can draw from the statement $\Leb(\Sigma \setminus \CZ) = 0$ and Kotani theory is that the a.c.\ spectrum of $\CE_\omega$ is $\Leb(\Sigma \setminus \sigma_\ac(\CE_\omega)) = 0$ for a.e.\ $\omega$ (in the almost-periodic setting, one can strengthen this to $\Leb(\Sigma\setminus \sigma_\ac(\CE^\omega)) = 0$ for every $\omega \in \Omega$ using \cite[Theorem~10.9.11]{S2}). In any event, one cannot exclude singular spectrum without additional arguments beyond Kotani theory.
\smallskip
\end{itemize}

\subsection{Motivation: Quantum Walks} An important motivation to study CMV matrices comes from quantum walks, which we briefly describe. Quantum walks are quantum mechanical analogues of classical random walks. A good review of the topic can be found in \cite{QWcomp}. They are important in physics and computer science, particularly given their possible applications in quantum computing algorithms. For instance, quantum walks have been applied to the element distinctiveness problem \cite{Ambainis}, universal computation \cite{Childs}, and search algorithms \cite{SKW}. A good summary of the uses of quantum walks in quantum computing algorithms can be found in \cite{Venegas-Andraca}. There has been a lot of recent interest in using spectral theory to determine the spreading rates and other quantum dynamical characteristics of certain quantum walk models.

A quantum walk on $\Z$ is given by the iteration of a unitary operator on $\mathcal{H} := \ell^2(\Z) \otimes \C^2$. The $\ell^2(\mathbb Z)$ component represents the wavefunction of the walker (or, alternatively, the probability amplitude for the walker's location) and the $\mathbb C^2$ term represents the walker's spin. Denoting the standard basis of $\C^2$ by $\{\vec e_+, \vec e_-\}$, we obtain an orthonormal basis for $\mathcal H$ by taking vectors of the form $\delta_n^+ = \delta_n \otimes \vec e_+$ and $\delta_n^- = \delta_n \otimes \vec e_-$. Then, the time-one transition operator of a quantum walk is given by $\mathbf{U} = \mathbf{S}\mathbf{Q}$, where $\mathbf{S}$ is the biased shift $\delta_n^\pm \mapsto \delta_{n \pm 1}^\pm$, and $\mathbf{Q}$ is of the form
\[
\mathbf{Q} = \bigoplus_{n \in \Z} (\Id_n \otimes Q_n),
\]
where $Q_n \in \U(2)$ and $\Id_n$ denotes the identity operator on coordinate $n$, where one views
\[
\ell^2(\Z)
\cong
\bigoplus_{n \in \Z} \C.
\]
In the important paper \cite{CGMV}, Cantero, Gr\"unbaum, Moral, and Vel\'azquez show that the unitary operator $\mathbf{U}$ enjoys a matrix representation as a CMV matrix if one orders the orthonormal basis suitably (namely, in the order $\ldots \delta_0^+,\delta_0^-,\delta_1^+,\delta_1^-,\delta_2^+,\ldots$). This point of view has been quite fruitful in the analysis of 1D coined quantum walks, as it enables one to use theorems about  CMV matrices to draw conclusions about the spreading behavior of the quantum walk.

As an example, using Theorem A.2 of \cite{FilOng2017JFA} we can see that our Theorem \ref{t:PTac} has a direct consequence for the speading of quantum  walks:
\begin{coro}
Suppose the coins $Q_n$ of the quantum walk $\mathbf{U}$ lie in the Pastur-Tkachenko class. Then $\mathbf{U}$ has purely a.c.\ spectrum and hence exhibits scattering in the sense that any wave packet leaves any compact region in finite time. That is, for any $J\in\Z_+$, and $\psi\in\ell^2(\mathbb Z)$,
\begin{equation}\label{eq:RAGEac}
\lim_{n\to\infty}  
\sum_{j=-J}^J
\left\vert \left< \delta_j^+, \mathbf{U}^n \psi\right>\right\vert^2
+ \left\vert \left< \delta_j^-, \mathbf{U}^n \psi\right>\right\vert^2
=
0.
\end{equation}
\end{coro}
Using other methods, (see \cite{F2017CMP}) it is already known that these quantum walks actually exhibit ballistic transport. Nevertheless, this corollary demonstrates that our results have relevance to physics.

\section*{Acknowledgements}

J.F.\ was supported in part by an AMS--Simons travel grant, 2016--2018. T.V.\ was supported in part by NSF Grant DMS--1361625.

\section{Spectral approximation of ergodic CMV matrices}

Let us start off by defining a version of the GZ formalism. We remark that in our paper, we follow Simon's conventions \cite{S1,S2} regarding how $\CE$ depends on the $\alpha_n$'s. Note that these conventions differ from the notation used in \cite{GZ2006JAT}. We will use that any CMV operator $\CE$ enjoys a factorization into direct sums of $2\times 2$ unitaries of the form
\[
\Theta(\alpha)
\eqdef
\begin{bmatrix}
\overline{\alpha} & \rho \\
\rho &-\alpha
\end{bmatrix},
\quad
\rho = \rho_\alpha
\eqdef
\left(1-|\alpha|^2\right)^{1/2}.
\]
That is, we have
\begin{equation}\label{eq:LM} 
\CE = \CL\CM,
\end{equation} where
\[
\CL = \bigoplus_{j \in \Z} \Theta(\alpha_{2j}),
\quad
\CM = \bigoplus_{j \in \Z} \Theta(\alpha_{2j+1}),
\]
and the $\Theta$ matrix corresponding to $\alpha_n$ acts on coordinates $n$ and $n+1$. Given a solution $u$ of $\CE u = zu$, we define $v = \mathcal L^{-1} u$. One can check that 

\begin{equation}\label{eq:Mu=zv}\mathcal M u = zv,
\end{equation} and hence $\CE^\top v= zv$ (since each $\Theta(\alpha)$ is real-symmetric). Using the equation $u = \CL v$, we have
\[
\begin{bmatrix}
u(2j) \\ u(2j+1)
\end{bmatrix}
=
\Theta(\alpha_{2j})
\begin{bmatrix}
v(2j) \\ v(2j+1)
\end{bmatrix},
\]
which can be rearranged to yield
\[
\begin{bmatrix}
u(2j+1) \\ v(2j+1)
\end{bmatrix}
=
\frac{1}{\rho_{2j}}
\begin{bmatrix}
-\alpha_{2j} & 1 \\
1 & -\overline{\alpha_{2j}}
\end{bmatrix}
\begin{bmatrix}
u(2j) \\ v(2j)
\end{bmatrix}.
\]
Similarly, using $\CM u = zv$, we can deduce
\[
\begin{bmatrix}
zv(2j-1) \\ zv(2j)
\end{bmatrix}
=
\Theta(\alpha_{2j-1})
\begin{bmatrix}
u(2j-1) \\ u(2j)
\end{bmatrix},
\]
and get
\begin{equation}\label{eq:gzOddToEven}
\begin{bmatrix}
u(2j) \\ v(2j)
\end{bmatrix}
=
\frac{1}{\rho_{2j-1}}
\begin{bmatrix}
-\overline{\alpha_{2j-1}} & z \\
z^{-1} & -\alpha_{2j-1}
\end{bmatrix}
\begin{bmatrix} u(2j-1) \\ v(2j-1) \end{bmatrix}.
\end{equation}

Thus, we have
\[
\begin{bmatrix}
u(n+1) \\ v(n+1)
\end{bmatrix}
=
Y(n,z) \begin{bmatrix}
u(n) \\ v(n)
\end{bmatrix},
\]
where
\[
Y(n,z)
=
\frac{1}{\rho_n}
\begin{cases}
\begin{bmatrix}
-\alpha_n & 1 \\ 1 & - \overline{\alpha_n}
\end{bmatrix}
&  n \text{ is even,} \\[18pt]
\begin{bmatrix}
-\overline{\alpha_n} & z \\ z^{-1} & -\alpha_n
\end{bmatrix}
&  n \text{ is odd.}
\end{cases}
\]
For later use, we notice that \eqref{eq:gzOddToEven} can also  be inverted to yield
\begin{equation} \label{eq:GZMAGIC}
\rho_{2j-1}
\begin{bmatrix}
u(2j-1) \\ v(2j-1)
\end{bmatrix}
=
\begin{bmatrix}
\alpha_{2j-1} & z \\ z^{-1} & \overline{\alpha_{2j-1}}
\end{bmatrix}
\begin{bmatrix}
u(2j) \\ v(2j)
\end{bmatrix}.
\end{equation}

The first technical lemma is a formula for the derivative of the $n$th band function in terms of associated Bloch wave solutions. We first introduce suitable truncations of $\CE$ whose eigenvectors can be used to generate Bloch waves. Suppose that $\CE$ is $q$-periodic; throughout this section, we also assume that $q$ is even. We define $\CL_q = \CL_q(k)$ and $\CM_q = \CM_q(k)$ as in \cite[Equation~(11.2.7)]{S2}, that is:
\[
\CL_q
=
\begin{bmatrix}
\Theta_0 \\
& \ddots \\
&& \Theta_{q-2}
\end{bmatrix},
\quad
\CM_q
=
\CM_q(k)
=
\begin{bmatrix}
-\alpha_{q-1} &&&& \rho_{q-1} e^{-ikq} \\
 & \Theta_1 \\
 && \ddots \\
 &&& \Theta_{q-3} \\
 \rho_{q-1} e^{ikq} &&&& \overline{\alpha}_{q-1}
\end{bmatrix}.
\]
Then, we define $\CE_q = \CL_q \CM_q$ and the dual operator $\widetilde{\CE}_q = \CM_q \CL_q$\footnote{This is a subtle point. Since the full-line operators $\CL$ and $\CM$ are real-symmetric, one has $\CM \CL = \CE^\top$. However, this is no longer true for the Floquet operators, and hence one does not have $\widetilde{\CE}_q = \CE_q^\top$.} and let $\{z_n(k)\}_{n=1}^q$ denote an enumeration of the eigenvalues of $\CE_q$. Observe that $\CE_q$ and $\widetilde{\CE}_q$ have the same set of eigenvalues. This is evident  by looking at the proof of Lemma 2.2 of \cite{GZ2006JAT}; one can also see that they are unitarily equivalent viz.\ $\widetilde{\CE}_q = \CL_q^* \CE_q \CL_q$.  The eigenvectors of $\CE_q$ generate Bloch wave solutions to the difference equation $\CE u = zu$. Concretely, let $u_n = u_n(k)$ denote a normalized eigenvector of $\CE_q$ corresponding to the eigenvalue $z_n(k)$. We may extend $u_n(k,j)$ to all $j \in \Z$, obtaining a solution of $\CE u = z_n(k) u$ with
\[
u_n(k,j+q)
=
e^{ikq} u_n(k,j)
\text{ for all } j \in \Z.
\]
We then define $v_n(k) = \CL_q^{-1} u_n(k)$. One can check that $v_n$ solves the dual equation $\widetilde{\CE}_q v_n = z_n v_n$. As with $u$, we can extend $v_n$ to a globally defined solution of $\CE^\top v = z_n v$ with 
\[
v_n(k,j+q) = e^{ikq} v_n(k,j).
\]

\begin{lemma}
Given $k \in (0, \pi/q)$, let $u_n$ and $v_n$ be the solutions above, normalized so that
\begin{equation}\label{eq:uv-normalize}
\sum_{j=0}^{q-1}|u_n(k,j)|^2
=
\sum_{j=0}^{q-1}|v_n(k,j)|^2
=
1.
\end{equation}
Then for every integer $n$,
\begin{equation} \label{eq:dE/dk}
\dfrac{\dd z_n}{\dd k}(k)
=
iq \rho_{-1} [\overline{v_n(k,-1)}u_n(k,0) - \overline{v_n(k,0)}u_n(k,-1)].
\end{equation}
\end{lemma}

\begin{proof}
By (11.2.6) and (11.2.7) of \cite{S2}, we have 
\begin{equation}
\frac{\dd \CE_q}{\dd k}(k)
=
\mathcal L_q \frac{\dd \CM_q}{\dd k},
\end{equation}
and we may note that
\[
\frac{\dd \CM_q}{\dd k}
=
iq\begin{bmatrix}
0&\cdots&0& -\rho_{q-1}e^{-ikq}\\
\vdots&&&0\\
0&&&\vdots\\
\rho_{q-1}e^{ikq}&0&\cdots &0
\end{bmatrix} 
\]
Notice that for $k\in (0,\pi/q)$, $z_n(k)$ is nondegenerate by \cite[Theorem~11.2.2]{S2}. By \eqref{eq:Mu=zv} and the definitions, we have
\[
z_n(k)
=
\left\langle v_n(k), \CM_q(k) u_n(k) \right\rangle.
\]
We claim that
\begin{equation} \label{eq:eigDerivClaim}
\frac{\dd z_n}{\dd k}(k)
=
\left< v_n(k), \frac{\dd \CM_q}{\dd k}u_n(k)\right>.
\end{equation}
To see this, it suffices to show that
\begin{equation} \label{eq:cancellingTerms}
\left< \frac{\dd v_n}{\dd k}(k),  \CM_q(k)u_n(k)\right>
+ \left< v_n(k) ,  \CM_q(k) \frac{\dd u_n}{\dd k}(k)\right>
= 0.
\end{equation}
However, this is immediate from \eqref{eq:uv-normalize} and the unitarity of $\CL_q$, $\CM_q$. In particular, using a dot to denote differentiation with respect to $k$, we have
\begin{align*}
\left\langle \dot{v}_n, \CM_q u_n \right\rangle 
+ \left\langle v_n, \CM_q \dot{u}_n \right\rangle
& =
\left\langle \dot{v}_n, zv_n \right\rangle + \left\langle z^{-1} u_n, \dot{u}_n \right\rangle \\[3pt]
& = 
z \left( \left\langle \dot{v}_n, v_n \right\rangle + \left\langle  u_n, \dot{u}_n \right\rangle \right) \\[3pt]
& =
z \left( \left\langle \dot{v}_n, v_n \right\rangle + \left\langle \CL_q^*  u_n, \CL_q^* \dot{u}_n \right\rangle \right) \\[3pt]
& =
z \left( \left\langle \dot{v}_n, v_n \right\rangle + \left\langle v_n, \dot{v}_n \right\rangle \right) \\
& =
z \frac{\dd}{\dd k}\|v_n\|^2 \\[3pt]
& =
z \frac{\dd}{\dd k}(1) \\[3pt]
& = 0,
\end{align*}
where we have used $z = z_n \in \pD$, $u_n = \CL v_n$, and that $\CL_q$ is independent of $k$. Thus, \eqref{eq:cancellingTerms} follows, so we get \eqref{eq:eigDerivClaim}. Consequently,
\[
\frac{\dd z_n}{\dd k}(k)
=
iq[e^{ikq}\rho_{q-1}\overline{ v_n(k,q-1)}u_n(k,0)
-e^{-ikq}\rho_{q-1}\overline{ v_n(k,0)}u_n(k,q-1)]
\]
Using $u_n(k,j+q) = e^{ikq}u_n(k,j)$ and $v_n(k,j+q) = e^{ikq}v_n(k,j)$ we obtain
\[
\frac{\dd z_n}{\dd k}(k)
=
iq[\rho_{q-1} \overline{v_n(k,-1)}u_n(k,0)
- \rho_{q-1} \overline{v_n(k,0)}u_n(k,-1)],
\]
as desired.
\end{proof}

\begin{lemma} \label{l:monodromyNorm}
Suppose $\CE$ is $q$-periodic with $q$ even, and denote the associated monodromy matrix by
\[
\Phi_q(z)
=
Y(q-1,z) \cdots Y(0,z).
\]
For all $z$ with $\tr(\Phi_q(z)) \in (-2,2)$, we have
\[
\Vert \Phi_q (z)\Vert 
\leq
4 q \left \vert \dfrac{\dd z}{\dd k}\right\vert^{-1},
\]
where $k$ is the appropriate Bloch wave number.
\end{lemma}

\begin{proof}
Through a calculation virtually identical to that which proves \cite[Equation~(3.1)]{Last1993CMP}, we obtain
\begin{equation}\label{eq:PhiBound}
\Vert \Phi_q(z)\Vert^2
\leq
\frac{4}{1-\vert \left< x^+,x^-\right>\vert^2},
\end{equation}
where $x^\pm$ denote the normalized eigenvectors of $\Phi_q(z)$.

Here, \cite[Section~10.4]{S2} will prove useful. In particular, it is easy to check that $\Phi_q(z)$ is an element of the group $\SU(1,1)$. Furthermore, since we are concerned about $z$ in the interiors of bands, we are in the elliptic setting, where the trace of $\Phi_q(z)$ lies in $(-2,2)$. Using Theorem 10.4.3(a) and (10.4.16) of \cite{S2} we can see that if we set $x^+ = (u_n(k,0), v_n(k,0))^\top$ we can set $x^-=(\overline{v_n(k,0)}, \overline{u_n(k,0)})^\top$. Normalize the Bloch-wave solutions $u_n$ and $v_n$ such that $\|x^+\|^2 = \|x^-\|^2 = 1$. From this, we obtain
\begin{align*}
1 - |\langle x^+, x^-\rangle|^2
& =
1 - \left\vert \mathrm{Re}\left(2u_n(k,0)\overline{ v_n(k,0)}\right) \right\vert^2 \\
& =
(\vert u_n(k,0)\vert^2+\vert v_n(k,0)\vert^2)^2-\vert \mathrm{Re}(2u_n(k,0)\overline{ v_n(k,0)}) \vert^2 \\
& \geq
(\vert u_n(k,0)\vert^2+\vert v_n(k,0)\vert^2)^2-\vert 2u_n(k,0)\overline{ v_n(k,0)} \vert^2 \\
& =
(\vert u_n(k,0)\vert^2-\vert v_n(k,0)\vert^2)^2.
\end{align*}
Combining this with \eqref{eq:PhiBound}, we get
\begin{align}
\label{eq:prelimnormbd}
\Vert \Phi_q (z)\Vert 
\leq 
\frac{2}{\big| \vert u_n(k,0)\vert^2-\vert v_n(k,0)\vert^2 \big|}.
\end{align}

Because we have normalized $u_n$ and $v_n$ by taking $\|x^+\|^2 = |u_n(k,0)|^2 + |v_n(k,0)|^2 = 1$, we must have for all $n$
\begin{align}
\label{eq:maxentry}
\max\{|u_n(k,0)|^2, |v_n(k,0)|^2\} 
&\geq 
\frac{1}{2}.
\end{align}
Furthermore, because $u_n = \CL_q v_n$ are unitarily related, the vectors $u_n$ and $v_n$ have the same norm in $\C^q$; by \eqref{eq:maxentry}, this norm has a uniform lower bound
\begin{align*}
\|u_n\|^2 
= 
\|v_n\|^2 
=: 
N^2 
\geq 
\frac{1}{2}.
\end{align*} 
Consider now the vectors $\widetilde{u}_n := u_n/N$, $\widetilde{v}_m := v_n/N$ (which are unit vectors in $\C^q$). By \eqref{eq:GZMAGIC}, we get
\begin{align*}
\rho_{-1}\widetilde{u}_n(k,-1)
& = 
z\widetilde{v}(k,0) + \alpha_{-1} \widetilde{u}(k,0) \\
\rho_{-1} \widetilde{v}_n(k,-1)
& =
z^{-1} \widetilde{u}(k,0) + \overline{\alpha_{-1}} \widetilde{v}(k,0).
\end{align*}
Plugging this into \eqref{eq:dE/dk} we obtain
\begin{align}
\nonumber
\frac{\dd z_n(k)}{\dd k}
& =
iq \rho_{-1}\left( \overline{\widetilde{v}_n(k,-1)} \widetilde{u}_n(k,0) - \overline{\widetilde{v}_n(k,0)} \widetilde{u}_n(k,-1) \right) \\
\nonumber
& =
iq z \left(\vert \widetilde{u}_n(k,0)\vert^2-\vert \widetilde{v}_n(k,0)\vert^2 \right) \\
&= 
\frac{iq z}{N^2}\left(\vert u_n(k,0)\vert^2-\vert v_n(k,0)\vert^2 \right). \label{eq:dE/dk00}
\end{align}

Using $N^2 \geq 1/2$ and $|z| = 1$, it follows from \eqref{eq:prelimnormbd} and \eqref{eq:dE/dk00} that

\begin{equation}
\Vert \Phi_q (z)\Vert\leq \frac{4q}{\left\vert\dfrac{\dd z_n}{\dd k}(k) \right\vert}.
\end{equation}
\end{proof}

With Lemma~\ref{l:monodromyNorm} in hand, the main technical challenges have been dealt with. At this point, one can prove Theorem~\ref{t:perapprox} in the same way that Last proves \cite[Theorem~1]{Last1993CMP}. We provide a short sketch for the reader's benefit. Please consult Section 11.2 of \cite{S2} for a helpful discussion of Floquet theory for CMV operators.
\begin{proof}[Proof Sketch of Theorem~\ref{t:perapprox}]
For each $\omega \in \Omega$, define the sets
\[
S_\omega
:=
\limsup_{m \to \infty} \set{z \in \Int \, \Sigma^{\omega,m} : \left| \frac{\dd z}{\dd k} \right| \geq \frac{1}{m^2} },
\quad 
\A_\omega 
:= 
\limsup_{p \to \infty} \Sigma^{\omega,p}
\]
By the Borel--Cantelli lemma, one has
\[
\Leb(\A_\omega \setminus S_\omega) = 0
\]
for every $\omega \in \Omega$. On the other hand, by Fubini's theorem, the multiplicative ergodic theorem, and Lemma~\ref{l:monodromyNorm}, one has $\Leb(S_\omega \setminus \CZ) = 0$ for almost every $\omega$. It follows that $\Leb(\A_\omega \setminus \CZ) = 0$ for $\mu$-a.e.\ $\omega$, as desired.
\end{proof}

\section{Spectral approximation of limit-periodic CMV matrices}

In this section, we will collect a handful of facts about the spectra of limit-periodic CMV matrices and their periodic approximants.  Throughout, for a subset $E \subset \partial\D$, we denote
\begin{align*}
B_\varepsilon(E) := \{z \in \partial\D : \inf_{x \in E} |z-x| < \varepsilon\}.
\end{align*}
The \emph{Hausdorff distance} between two compact sets $F,K \subset \partial \D$ is defined by
\[
d_\Hd(F,K)
=
\inf\set{\e > 0 : F \subset B_\e(K) \text{ and } K \subset B_\e(F)}.
\] 
The fundamental fact driving the analysis in this section is the following
\begin{lemma}
\label{lem:hausleqop}
For any unitary operators $U,V$ on $\ell^2$, we have
\begin{align}
\label{eq:hausleqop}
d_\Hd(\sigma(U), \sigma(V)) 
\leq 
\|U-V\|,
\end{align}
where $\|\cdot\|$ denotes the usual operator norm.
\end{lemma}
\begin{proof}
This follows from \cite[Theorem~V.4.10]{Kato} Kato states the theorem for self-adjoint operators, but the same proof works for unitary (or even bounded, normal) operators. Alternatively, there is an explicit proof for CMV operators in \cite[Proposition~4]{Ong12}.
\end{proof}
By way of Lemma \ref{lem:hausleqop}, we can use the rate function of a limit-periodic operator to control the rate of spectral convergence of the periodic approximants.

The following lemmata are classical, but their proofs are short, so we reproduce them for the sake of completeness:
\begin{lemma}
\label{lem:limsupcont}
If $\{\Sigma_n\}_{n\geq 1}$ and $\Sigma$ are compact sets in $\partial\D$ such that $d_\Hd(\Sigma,\Sigma_n) \to 0$, then $\limsup \Sigma_n \subset \Sigma$.
\end{lemma}

\begin{proof}
Arguing by contraposition, suppose $t \not\in \Sigma$.  Because $\Sigma$ is compact, there exists $\e > 0$ such that $\inf \{|t-x| : x \in \Sigma\} \geq \e$.  However, by assumption, there also exists an $N$ such that, for all $n \geq N$, $d_\Hd(\Sigma,\Sigma_n) < \e$, and thus $\Sigma_n \subset B_{\e}(\Sigma)$ for such $n$. Consequently, $t \notin \Sigma_n$ for $n \ge N$, which means
\[
t \notin 
\limsup_{n \to \infty} \Sigma_n.
\]
\end{proof}

Recall that for the special almost-periodic class of ergodic CMV matrices, the underlying probability space $\Omega$ is a compact monothetic group with translation $T$ by a topological generator.  In this case the action of the transformation $T$ manifests as a shift on the Verblunsky coefficients $\alpha$; since $T$ is a topological generator for $\Omega$, we call $\Omega$ the ``shift-hull" of the almost-periodic CMV matrix $\CE$ in this case.

\begin{lemma}
\label{lem:APspectra}
Let $\CE$ be an almost-periodic CMV matrix with shift-hull $\Omega$.  Then for $\omega, \omega_0 \in \Omega$, we have $\sigma(\CE^\omega) = \sigma(\CE^{\omega_0})$.
\end{lemma}

\begin{proof}
Given $\omega \in \Omega$, by almost-periodicity there exists a sequence $(n_j)_{j \in \N}$ such that $T^{n_j}\omega_0 =: \omega_j \to \omega$ as $j \to \infty$, and in particular $\slim_{j \to \infty} \CE^{\omega_j} = \CE^\omega$. Since $\CE^{\omega_j}$ is unitarily equivalent to $\CE^{\omega_0}$ for each $j$, one has
\begin{align*}
\sigma(\CE^{\omega}) 
\subset
\overline{\liminf _{j\to\infty}\sigma(\CE^{\omega_j})} 
=
\sigma(\CE^{\omega_0}).
\end{align*}
The opposite inclusion follows by symmetry.
\end{proof}

\begin{prop}
\label{pr:LPsetdif0}
Let $\CE$ be a limit-periodic CMV matrix with $q_n$-periodic approximants $\CE_n$. Denote $\Sigma = \sigma(\CE)$ and $\Sigma_n = \sigma(\CE_n)$.  If
\begin{align}
\label{eq:LPwellapprox}
\lim_{n \to \infty} q_n \|\CE - \CE_n\| = 0,
\end{align}
then $\Leb(\Sigma\setminus \Sigma_n) \to 0$ as $n \to \infty$.
\end{prop}
\begin{proof}
Denote by $\e_n := \|\CE - \CE_n\|$.  Then by \eqref{eq:hausleqop} and the definition of the Hausdorff metric, we have
\begin{align}
\label{eq:LPwaproof}
\Sigma \setminus \Sigma_n \subset B_{\e_n}(\Sigma_n) \setminus \Sigma_n.
\end{align}
Because $\Sigma_n$ has at most $q_n$ bands, the right hand side of (\ref{eq:LPwaproof}) has at most $2q_n$ connected components, each of length at most $\e_n$ by \eqref{eq:hausleqop}.  Thus,
\begin{align*}
\Leb(\Sigma \setminus \Sigma_n)
\leq
2q_n\e_n,
\end{align*}
which tends to $0$ as $n \to \infty$ by the assumption \eqref{eq:LPwellapprox}.
\end{proof}

Using the same sort of methods, we can also address the proof of Theorem~\ref{t:LPposmeas}. The first part of the proof is made much easier if we use the ``sieving'' construction. Concretely, given a CMV operator $\CE$, let $\widehat{\CE}$ denote the CMV operator with
\begin{equation} \label{eq:sievedDef}
\widehat{\alpha}_{2j}
=
0,
\quad
\widehat{\alpha}_{2j-1} = \alpha_j,
\quad
j \in \Z.
\end{equation}
This induces a simple change in the spectrum; namely, if $E_2:\pD\to\pD$ denotes the two-fold cover $z\mapsto z^2$, then
\begin{equation} \label{eq:sieveSpec}
\sigma(\widehat{\CE})
=
E_2^{-1}(\sigma(\CE)).
\end{equation}
In other words, one obtains $\sigma(\widehat{\CE})$ by taking two scaled copies of $\sigma(\CE)$ and putting them on the top and bottom halves of $\pD$.  To see this, one can verify by hand that $\widehat\CE^2 \cong \CE \oplus \CE^\top$. The calculation is known to experts, but may not be obvious to the uninitiated, so we will sketch the outline for the reader's convenience. First, let $\widehat\CE = \widehat\CL\widehat\CM$ denote the factorization of $\widehat\CE$ as in \eqref{eq:LM}. Then, straightforward calculations using the definitions yield
\[
\widehat\CL\delta_{2j}
=
\delta_{2j+1},
\quad
\widehat\CL\delta_{2j+1}
=
\delta_{2j}
\]
and
\[
\widehat\CM \delta_{2j-1}
=
\overline{\alpha_j}\delta_{2j-1} + \rho_j \delta_{2j},
\quad
\widehat\CM\delta_{2j}
=
\rho_j \delta_{2j-1} - \alpha_j \delta_{2j}.
\]
Therefore, one can verify that
\begin{align}
\label{eq:sievebasis1}
\widehat\CE^2 \delta_{4n-1}
& =
\overline{\alpha_{2n}} \rho_{2n-1} \delta_{4n-4} - \overline{\alpha_{2n}}\alpha_{2n-1}\delta_{4n-1} +\overline{\alpha_{2n+1}}\rho_{2n} \delta_{4n} + \rho_{2n+1} \rho_{2n}\delta_{4n+3} \\
\label{eq:sievebasis2}
\widehat\CE^2 \delta_{4n}
& =
\rho_{2n} \rho_{2n-1} \delta_{4n-4} - \rho_{2n}\alpha_{2n-1}\delta_{4n-1} - \overline{\alpha_{2n+1}}\alpha_{2n} \delta_{4n} - \rho_{2n+1}\alpha_{2n}\delta_{4n+3} \\
\label{eq:sievebasis3}
\widehat\CE^2 \delta_{4n+1}
& =
\overline{\alpha_{2n+1}} \rho_{2n} \delta_{4n-2} - \overline{\alpha_{2n+1}}\alpha_{2n} \delta_{4n+1} + \overline{\alpha_{2n+2}}\rho_{2n+1} \delta_{4n+2} + \rho_{2n+2}\rho_{2n+1} \delta_{4n+5} \\
\label{eq:sievebasis4}
\widehat\CE^2 \delta_{4n+2}
& =
\rho_{2n+1} \rho_{2n} \delta_{4n-2} - \rho_{2n+1}\alpha_{2n} \delta_{4n+1} - \overline{\alpha_{2n+2}\alpha_{2n+1}} \delta_{4n+2} - \rho_{2n+2}\alpha_{2n+1} \delta_{4n+5}.
\end{align}
Defining subspaces
\begin{align*}
\mathcal{X}
& =
\ell^2(\set{k \in \Z: k \equiv 0 \mod 4 \text{ or } k \equiv 3 \mod 4}) \\
\mathcal{Y}
& =
\ell^2(\set{k \in \Z : k \equiv 1 \mod 4 \text{ or } k \equiv 2 \mod 4}),
\end{align*}
the calculations in Eqs.~\eqref{eq:sievebasis1}--\eqref{eq:sievebasis4} show that $\widehat\CE^2$ leaves $\mathcal{X}$ and $\mathcal{Y}$ invariant and that $\widehat\CE^2|_{\mathcal{X}} \cong \CE$ and $\widehat\CE^2|_{\mathcal{Y}} \cong \CE^\top$. In particular, the claim about the spectrum holds. Moreover, we see that $\CE$ has purely a.c.\ spectrum if and only if $\widehat{\CE}$ has purely a.c.\ spectrum.

Additionally, notice that the Szeg\H{o} matrices (defined in \eqref{eq:szegotmdef}) obey
\begin{equation} \label{eq:sieveTMs}
S(\alpha,z)S(0,z)
=
S(\alpha,z^2).
\end{equation}
In particular, since the spectrum of $\CE$ is given by the closure of the set of $z \in \pD$ at which the Szeg\H{o} recursion enjoys a polynomially bounded solution \cite{DFLY2016DCDS}, \eqref{eq:sieveTMs} also suffices to establish \eqref{eq:sieveSpec}. Beyond that, \eqref{eq:sieveTMs} clearly implies $\widehat{L}(z) = L(z^2)$ for $z \in \pD$, where $\widehat L$ denotes the Lyapunov exponent corresponding to $\widehat{\CE}$. Consequently, since $\Leb(E_2^{-1}(S)) = \Leb(S)$ for every set $S$, one has
\[
\Leb(\Sigma\setminus\CZ)
=
\Leb(\widehat{\Sigma}\setminus\widehat{\CZ}),
\]
where we have used hats to denote the sets associated to the sieved CMV operators. The outcome of this discussion is that it suffices to work with the sieved CMV operators, and hence, one may as well assume that all even Verblunsky coefficients vanish.

This is quite helpful, because one can easily prove operator inequalities in terms of Verblunsky coefficients in the sieved setting. In particular, we have the following lemma:

\begin{lemma} \label{l:CMVdiffEst}
 If $\CE$ and $\CE'$ are CMV operators with coefficient sequences $\alpha$ and $\alpha'$ such that $\alpha_{2j} = \alpha_{2j}' = 0$ for every $j$ and $\|\alpha\|_\infty, \|\alpha'\|_\infty < 1$, then
\[
\|\alpha - \alpha'\|_\infty
\leq
\| \CE - \CE'\|
\leq
C \|\alpha - \alpha'\|_\infty,
\]
where $C$ is a constant that depends only on $\max(\|\alpha\|_\infty, \|\alpha'\|_\infty)$.
\end{lemma}

\begin{proof}
The upper bound is well-known; compare \cite[Equation~(4.3.13)]{S2}. For the lower bound, begin by using \eqref{eq:LM}, to observe that
\[
\CL = \CL'=
\bigoplus_{j \in \Z} \Theta(0),
\]
and thus
\[
\|\CE - \CE'\|
=
\| \CM - \CM'\|
=
\sup_j \|\Theta(\alpha_{2j-1}) - \Theta(\alpha_{2j-1}') \|
\geq
\|\alpha - \alpha'\|_\infty.
\]
\end{proof}

\begin{proof}[Proof of Theorem~\ref{t:LPposmeas}]
There are two claims.  First, we prove that, if \eqref{eq:LPwellapprox} holds, then one has $\Leb(\Sigma \setminus \CZ) = 0$. To that end, let $\CE$ be given with $\alpha$ satisfying \eqref{eq:alphao(q)}, let $\alpha^{(n)}$ denote the coefficients of $\CE_n$, and put $\delta_n = \| \CE_n - \CE\|$; by sieving, we may assume that $\alpha_{2j} = 0$ for every $j$ and that this holds for all the periodic approximants of $\CE$ as well.  Let us begin by defining $\Omega$ to be the shift-hull of $\alpha$, that is
\[
\Omega
=
\overline{\set{S^k\alpha : k \in \Z}},
\]
where $S$ denotes the left shift and the closure is taken in $\ell^\infty$. Then, for each $\omega \in \Omega$, we get a CMV operator $\CE^\omega$ simply by using $\omega$ as a coefficient sequence. 
In particular, $\CE = \CE^{\omega_0}$ with $\omega_0 = \alpha$. 

We wish to emphasize an important but subtle difference between $\CE_n$ and $\widetilde \CE^{\omega_0, q_n}$. These are both $q_n$-periodic approximations of $\CE$, but $\widetilde \CE^{\omega_0, q_n}$ is specifically defined by taking the first $q_n$ terms of the $\alpha$ sequence and then repeating them (see \eqref{eq:widetildealpha}), whereas $\CE_n$ has Verblunsky coefficient sequence $\alpha^{(n)}$.
We then have by Lemma~\ref{l:CMVdiffEst},
\begin{align*} 
\|\CE^{\omega_0} - \widetilde\CE^{\omega_0, q_n}\|
& \leq
C \|\alpha - \widetilde\alpha^{\omega_0,q_n}\|\\
& \leq
C( \|\alpha - \alpha^{(n)}\| + \|\alpha^{(n)} - \widetilde \alpha^{\omega_0, q_n}\| )\\
& \leq
2C \|\alpha - \alpha^{(n)}\|\\
& \leq
2C \delta_n
\end{align*}

Shifting, this holds with $\omega_0$ replaced by any $\omega \in \Omega$. Since $\Sigma^{\omega,q_n}$ has at most $q_n$ connected components, we deduce that
\[
\Leb(\Sigma \setminus \Sigma^{\omega,q_n})
\leq
4 C q_n \delta_n,
\]
and the right-hand side goes to zero as $n\to\infty$ by assumption. It follows that
\[
\Leb\left( \Sigma\setminus \limsup_{k\to\infty} \Sigma^{\omega,k} \right)
=
0
\]
for each $\omega$, and hence $\Leb(\Sigma \setminus \CZ) = 0$ by Theorem~\ref{t:perapprox}.
\medskip

Now we prove second claim: that \eqref{eq:LPsumbd} implies \eqref{eq:alphao(q)} and $\Leb(\Sigma) > 0$. However, the first part is immediate: if \eqref{eq:LPsumbd} holds, then \eqref{eq:alphao(q)} follows, since
\begin{align*}
q_n \|\alpha - \alpha^{(n)}\|
\leq
q_n \|\CE - \CE_n\|
\leq 
q_n \sum_{m = n+1}^\infty \|\CE_m - \CE_{m-1}\|
\leq 
\sum_{m = n+1}^\infty q_m\|\CE_m - \CE_{m-1}\|,
\end{align*}
which is the tail of a convergent series and hence converges to zero as $n\to\infty$, proving \eqref{eq:alphao(q)}.

Finally, we prove that $\Leb(\Sigma) > 0$ under the assumption \eqref{eq:LPsumbd}.  By Lemma \ref{lem:hausleqop} and the limit-periodicity of $\CE$, we have $d_\Hd(\Sigma,\Sigma_n) \to 0$ as $n \to \infty$. Furthermore, by Lemma \ref{lem:limsupcont} and the semicontinuity of measure, we have
\begin{align*}
\Leb(\Sigma) 
\geq 
\Leb(\limsup \Sigma_n) 
\geq 
\limsup_{n \to \infty} \Leb(\Sigma_n),
\end{align*}
so it suffices to show that $\limsup_{n \to \infty} \Leb(\Sigma_n) > 0$.

For each $n$, we have $\Leb(\Sigma_n) \geq \Leb(\Sigma_{n-1}) - \Leb(\Sigma_{n-1} \setminus \Sigma_n)$.  Inductively, one sees
\begin{align}
\label{eq:LPpmtelescope}
\Leb(\Sigma_n) 
\geq 
\Leb(\Sigma_k) 
- \left[\sum_{j=k+1}^n \Leb(\Sigma_{j-1} \setminus \Sigma_j) \right]
\end{align}
for all $n > k$. Following the proof of Proposition \ref{pr:LPsetdif0} (and, in particular, using a suitable version of equation \eqref{eq:LPwaproof}), we find that $\Leb(\Sigma_{j-1} \setminus \Sigma_j) \leq 2q_j \delta_j$. Using this together with Equation~\eqref{eq:LPpmtelescope}, we can conclude that
\begin{align*}
\Leb(\Sigma_n) 
\geq
\Leb(\Sigma_{k})- \sum_{j=k+1}^n 2q_j \delta_j
>
0,
\end{align*}
where we have used the assumption \eqref{eq:LPsumbd}.
\end{proof}

As a concluding remark to this section, we note that the natural analogue of Theorem \ref{t:LPposmeas} holds for Jacobi and Schr\"odinger operators with precisely the same proof.

\section{Well-approximated limit-periodic CMV matrices are reflectionless}

We now apply our Theorem \ref{t:LPposmeas} to prove pure absolute continuity of the spectrum for Pastur-Tkachenko class CMV matrices.

One common thread relating extended CMV matrices to the apparently quite different class of Schr\"odinger and Jacobi operators on the whole line is Weyl-Titchmarsh theory, by which one studies the whole-line operator $\CE$ via its cyclic restrictions $\CE_{\pm,k}$ to the half-lines $\Z_{-,k}:=\Z \cap (-\infty, k]$ and $\Z_{+,k}:=\Z \cap [k,\infty)$ \cite{GZ2006JAT}.  Critical to this study are the Weyl-Titchmarsh coefficients, defined for $z \in \C \setminus \partial\D$ by
\begin{align*}
m_\pm(z,k) &= \pm\langle \delta_k, (\CE_{\pm,k} + zI)(\CE_{\pm,k} - zI)^{-1}\delta_k \rangle_{\ell^2(\Z_{\pm,k})}, \\
M_+(z,k) &= m_+(z,k-1), \\
M_-(z,k) &= \frac{\Re(1-\overline{\alpha_k})+i\Im(1+\overline{\alpha_k})m_-(z,k-2)}{i\Im(1-\overline{\alpha_k}) + \Re(1+\overline{\alpha_k})m_-(z,k-2)}.
\end{align*}
Restricted to the unit disk, the coefficients $\pm M_\pm$ are Caratheodory functions; that is, functions holomorphic from $\D$ to the right half-plane $\{z \in \C : \Re(z) > 0\}$.  Consequently, the functions $M_\pm$ have well-defined radial limits at $e^{i\theta}$ for (Lebesgue) almost-every $\theta \in [0,2\pi)$.  We denote these limits
\begin{align*}
M_\pm(e^{i\theta},k) := \lim_{r \uparrow 1} M_\pm(re^{i\theta},k)
\end{align*}
when they exist.

We say that $\CE$ is \emph{reflectionless} when
\begin{align}
\label{eq:reflless}
\text{ for every } k \in \Z, \; M_+(e^{i\theta},k) &= -\overline{M_-(e^{i\theta},k)} \,\text{ for almost-every } e^{i\theta} \in \sigma(\CE).
\end{align}
By design, the reflectionless condition allows one to construct consistent analytic continuations of the Weyl-Titchmarsh coefficients beyond the unit disk.  This in turn implies a strong determinism between half-line restrictions of the CMV matrix.  The reflectionless property is intimately related to absolute continuity of the spectrum \cite{BRZ2009CMP, GZ2009JDE, REM11} via the following fundamental results:

\begin{theorem}[\mbox{\cite[Theorem~1.4]{BRZ2009CMP}}]
\label{t:brz}
Almost-periodic CMV matrices are reflectionless on $\CZ$.
\end{theorem}
\begin{theorem}[\mbox{\cite[Theorem~3.5]{GZ2009JDE}}]
\label{t:gz09}
If $\sigma(\CE)$ is homogeneous and $\CE$ is reflectionless thereupon, then $\CE$ has purely a.c.\ spectrum.
\end{theorem}

We will exploit these results in our proof of Theorem \ref{t:PTac}.  We first prove via our previous results that the a.c. spectrum is full; from there, we will use Theorem \ref{t:brz} to conclude reflectionlessness.  In fact, we have the following broader result, which is itself a straightforward application of Theorem \ref{t:LPposmeas}:

\begin{theorem}
\label{t:LPreflless}
If $\CE$ is a limit-periodic CMV matrix that satisfies the assumption \eqref{eq:LPsumbd} of Theorem~\ref{t:LPposmeas}, then $\Leb(\Sigma)>0$ and $\CE$ is reflectionless on $\Sigma$. 
\end{theorem}
\begin{proof}
By Theorem \ref{t:LPposmeas}, all that remains to be shown is that $\CE$ is reflectionless.  But this follows immediately from Theorem \ref{t:brz}.
\end{proof}

Theorem \ref{t:LPreflless} allows us to conclude that Pastur-Tkachenko class CMV matrices are reflectionless on their spectrum.  Previous spectral estimates from \cite{FilLuk2017JST} will provide that the spectrum is homogeneous, and we will finally appeal to Theorem \ref{t:gz09} to conclude purely a.c. spectrum.

We are now in position for the

\begin{proof}[Proof of Theorem \ref{t:PTac}]
PT class implies homogeneity of the spectrum by \cite{FilLuk2017JST}, and the rate function satisfies \eqref{eq:LPwellapprox} almost trivially. Thus, Theorem \ref{t:gz09} applies.
\end{proof}


\begin{thebibliography}{00}
\bibitem{Ambainis}A.\ Ambainis, Quantum walk algorithm for element distinctness, \textit{SIAM Journal on Computing},\ \textbf{37} No. 1 (2007), 210-239.

\bibitem{AK98} J.\ Asch, A.\ Knauf, Motion in periodic potentials, \textit{Nonlinearity} \textbf{11} (1998) 175--200.

\bibitem{Avila2009CMP} A.\ Avila, On the spectrum and Lyapunov exponent of limit periodic Schr\"odinger operators, \textit{Comm.\ Math.\ Phys.}\ \textbf{288} (2008) 907--918.

\bibitem{BRZ2009CMP} J.\ Breuer, E.\ Ryckman, M.\ Zinchenko, Right limits and reflectionless measures for CMV matrices, \textit{Comm.\ Math.\ Phys.}\ \textbf{292} (2009) 1--28.

\bibitem{CGMV} M.-J.\ Cantero, A.\ Gr\"unbaum, L.\ Moral, L.\ Vel\'azquez, Matrix-valued Szeg\H{o} polynomials and quantum random walks, \textit{Comm.\ Pure Appl.\ Math.}\ \textbf{63} (2010), 464--507.

\bibitem{Childs} A.\ M.\ Childs, Universal computation by quantum walk, \textit{Phys. Rev. Lett.}, \textbf{108} No. 18 (2009), 180501 

\bibitem{DamFilLuk2042JST} D.\ Damanik, J.\ Fillman, M.\ Lukic, Limit-periodic continuum Schr\"odinger operators with zero-measure Cantor spectrum, \textit{J.\ Spectral Th.}, in press.

\bibitem{DFLY2015IMRN} D.\ Damanik, J.\ Fillman, M.\ Lukic, W.\ Yessen, Uniform hyperbolicity for Szegő cocycles and applications to random CMV matrices and the Ising model, \textit{Int.\ Math.\ Res.\ Not.}\ \textbf{2015} (2015), 7110--7129.

\bibitem{DFLY2016DCDS} D.\ Damanik, J.\ Fillman, M.\ Lukic, W.\ Yessen, Characterizations of uniform hyperbolicity and spectra of CMV matrices, \textit{Discrete and Continuous Dynamical Systems -- Series S} \textbf{9} (2016), 1009--1023.

\bibitem{DamGor2016AdvM} D.\ Damanik, A.\ Gorodetski, An extension of the Kunz-Souillard approach to localization in one dimension and applications to almost-periodic Schr\"odinger operators, \textit{Adv.\ Math.}\ \textbf{297} (2016), 149--173.

\bibitem{DLY2015CMP} D.\ Damanik, M.\ Lukic, W.\ Yessen,  Quantum dynamics of periodic and limit-periodic Jacobi and block Jacobi matrices with applications to some quantum many body problems, \textit{Commun.\ Math.\ Phys.}\ \textbf{337} (2015), 1535--1561.

\bibitem{DMY2013JSP} D.\ Damanik, P.\ Munger, W.\ Yessen, Orthogonal polynomials on the unit circle with Fibonacci Verblunsky coefficients, II. Applications, \textit{J.\ Stat.\ Phys.}\ \textbf{153} (2013), 339--362.

\bibitem{Egorova1987} I.\ E.\ Egorova, Spectral analysis of Jacobi limit-periodic matrices, \textit{Dokl.\ Akad.\ Nauk Ukrain.\ SSR Ser.\ A} \textbf{3} (1987), 7--9. (in Russian)

\bibitem{F2017CMP} J.\ Fillman, Ballistic transport for limit-periodic Jacobi matrices with applications to quantum many-body problems, \textit{Commun.\ Math.\ Phys.}\ \textbf{350} (2017), 1275--1297.

\bibitem{FilLuk2017JST} J.\ Fillman, M.\ Lukic, Spectral homogeneity of limit-periodic Schr\"odinger operators, \textit{J.\ Spectral Th.}\ \textbf{7} (2017), 387--406.

\bibitem{FilOng2017JFA} J.\ Fillman, D.C.\ Ong, Purely singular continuous spectrum for limit-periodic CMV operators with applications to quantum walks, \textit{J.\ Funct.\ Anal.}\ \textbf{272} (2017), 5107--5143.

\bibitem{Gan2010} Z.\ Gan, An exposition of the connection between limit-periodic potentials and profinite groups, \textit{Math.\ Model.\ Nat.\ Phenom.}\ \textbf{5:4} (2010), 158--174.

\bibitem{GZ2006JAT} F.\ Gesztesy, M.\ Zinchenko, Weyl-Titchmarsh theory for CMV operators associated with orthogonal polynomials on the unit circle \textit{Journal of Approximation Theory} \textbf{139} (2006) 172--213.

\bibitem{GZ2009JDE} F.\ Gesztesy, M.\ Zinchenko, Local spectral properties of reflectionless Jacobi, CMV, and Schr\"odinger operators, \textit{J.\ Diff.\ Eq.}\ \textbf{246} (2009), 78--107.

\bibitem{Kachkovskiy2016CMP} I.\ Kachkovskiy, On transport properties of isotropic quasiperiodic $XY$ spin chains, \textit{Commun.\ Math.\ Phys.}, \textbf{345} (2016), 659--673.

\bibitem{Kato} T.\ Kato, \textit{Perturbation Theory for Linear Operators}, 2nd ed. Grundlehren der mathematischen Wissenschaften~\textbf{132}, Springer-Verlag, New York, 1976.

\bibitem{Last1992CMP} Y.\ Last, On the measure of gaps and spectra for discrete $1$D Schr\"odinger operators, \textit{Commun.\ Math.\ Phys.}\ \textbf{149} (1992), 347--360.

\bibitem{Last1993CMP} Y.\ Last, A relation between a.c.\ spectrum of ergodic Jacobi matrices and the spectra of periodic approximants, \textit{Commun.\ Math.\ Phys.}\ \textbf{151} (1993), 183--192.

\bibitem{Last1994CMP} Y.\ Last, Zero measure spectrum for the almost Mathieu operator, \textit{Commun.\ Math.\ Phys.}\ \textbf{164} (1994), 421--432.

\bibitem{Ong12} D.C.\ Ong, Limit-periodic Verblunsky coefficients for orthogonal
polynomials on the unit circle, \textit{J. Math. Anal. Appl.} \ \textbf{394} (2012) 633--644.

\bibitem{PasTka1984} L.\ Pastur, V.\ A.\ Tkachenko, On the spectral theory of the one-dimensional Schr\"odinger operator with limit-periodic potential (Russian), \textit{Dokl.\ Akad.\ Nauk SSSR} \textbf{279} (1984) 1050--1053.

\bibitem{PasTka1988}  L.\ Pastur, V.\ A.\ Tkachenko, Spectral theory of a class of one-dimensional Schr\"odinger operators with limit-periodic potentials (Russian), \textit{Trudy Moskov.\ Mat.\ Obshch.}\ \textbf{51} (1988) 114--168.

\bibitem{REM11} Christian Remling, The absolutely continuous spectrum of {J}acobi matrices, {\em Ann. of Math. (2)}, \textbf{174} (2011) 125--171.

\bibitem{S1} B.\ Simon, \textit{Orthogonal Polynomials on the Unit Circle. Part~1. Classical Theory}, Colloquium Publications, 54, American Mathematical Society, Providence (2005).

\bibitem{S2} B.\ Simon, \textit{Orthogonal Polynomials on the Unit Circle. Part~2. Spectral Theory}, Colloquium Publications, 54, American Mathematical Society, Providence (2005).

\bibitem{SKW} N.\ Shenvi, J.\ Kempe, and R.\ B.\ Whaley, A quantum walk search algorithm, \textit{Phys. Rev. A}\ \textbf{67} No. 5 (2003), 052307.

\bibitem{QWcomp} Venegas-Andraca, S.E., Quantum walks: a comprehensive review. Quantum Information Processing, 11.5 (2012), 1015-1106.

\bibitem{Venegas-Andraca} Venegas-Andraca, S.E., "Quantum walks for computer scientists." Synthesis Lectures on Quantum Computing 1.1 (2008), 1-119, Morgan \& Claypool

\bibitem{Zhao2017JDE} Z.\ Zhao, Ballistic transport in one-dimensional quasi-periodic continuous Schr\"odinger equation, \textit{J.\ Diff.\ Eq.}\ \textbf{262} (2017), 4523--4566.

\bibitem{ZhangZhao2017CMP} Z.\ Zhang, Z.\ Zhao, Ballistic Transport and Absolute Continuity of One-Frequency Schr\"odinger Operators, \textit{Commun.\ Math.\ Phys.}\ \textbf{351} (2017), 877--921.

\end{thebibliography}
\end{document}